\documentclass[reprint]{revtex4-2}

\usepackage{graphicx}
\usepackage[utf8]{inputenc}
\usepackage{natbib}
\usepackage{amsmath,mathtools,amssymb,amsfonts,dsfont,amsthm,mathrsfs,bm}

\DeclarePairedDelimiterX{\ket}[1]{\lvert}{\rangle}{#1}

\DeclarePairedDelimiterX{\bra}[1]{\langle}{\rvert}{#1}

\newcommand{\ketbra}[2]{\left\vert #1 \right\rangle \! \left\langle #2 \right\vert}

\newcommand{\tr}{\operatorname{Tr}} 

\newcommand{\sep}{\text{Sep}} 
\newcommand{\sepm}{\mathbf{SepM}} 

\newcommand{\cnot}{\text{CNOT}} 

\mathchardef\mhyphen="2D

\newtheorem{theorem}{Theorem}
\newtheorem{proposition}[theorem]{Proposition}
\newtheorem{lemma}[theorem]{lemma}
\theoremstyle{remark}

\begin{document}
	\title{Relation between Quantum Coherence and Quantum Entanglement in Quantum Measurements}
	\author{Ho-Joon Kim}\email{eneration@gmail.com}
	\affiliation{Department of Mathematics and Research Institute for Basic Sciences, Kyung Hee University, Seoul 02447, Korea}
	
	\author{Soojoon Lee}\email{level@khu.ac.kr}
	\affiliation{Department of Mathematics and Research Institute for Basic Sciences, Kyung Hee University, Seoul 02447, Korea}
    
	\begin{abstract}
		Quantum measurement is a class of quantum channels that sends quantum states to classical states. We set up resource theories of quantum coherence and quantum entanglement for quantum measurements and find relations between them. For this, we conceive a relative entropy type quantity to account for the quantum resources of quantum measurements. The quantum coherence of a quantum measurement can be converted into the entanglement in a bipartite quantum measurement through coherence non-generating transformations. Conversely, a quantum entanglement monotone of quantum measurements induces a quantum coherence monotone of quantum measurements. Our results confirm that the understanding on the link between quantum coherence and quantum entanglement is valid even for quantum measurements which do not generate any quantum resource.
	\end{abstract}

	\maketitle

\section{Introduction}
Quantum superposition or quantum coherence is at the heart of quantum theory; it is indispensable to describe quantum features such as the double-slit experiment. Distinct from the coherence of classical lights, quantum coherence of optical fields has been the main subject of quantum optics since the foundational works \cite{glauber1963QuantumTheoryOptical,glauber1963CoherentIncoherentStates,sudarshan1963EquivalenceSemiclassicalQuantum}. Quantum information science provided rigorous concepts and tools to explore quantum coherence of finite-dimensional systems as well as optical modes in the name of the quantum resource theory \cite{streltsov2017QuantumCoherenceResource,chitambar2019QuantumResourceTheories,kdwu2021ExperimentalProgressQuantum}. Quantum coherence has been studied for a fixed basis \cite{levi2014QuantitativeTheoryCoherent,baumgratz2014QuantifyingCoherence,winter2016OperationalResourceCoherence}, for subspaces \cite{aberg2006QuantifyingSuperposition}, for a set of linearly independent states \cite{killoran2016ConvertingNonclassicalityEntanglement,theurer2017ResourceTheorySuperposition}, or concerning an enlarged space for a quantum measurement \cite{bischof2019ResourceTheoryCoherence,bischof2021QuantifyingCoherenceRespect}. Quantum coherence is also investigated in the continuous variable systems related to the nonclassicality of light \cite{killoran2016ConvertingNonclassicalityEntanglement,kctan2017QuantifyingCoherenceCoherent}.

Quantum entanglement, the typical quantum correlation \cite{schroedinger1935GegenwaertigeSituationQuantenmechanik,schroedinger1936ProbabilityRelationsSeparated,horodecki2009QuantumEntanglement,piani2011AllNonclassicalCorrelations,gharibian2011CharacterizingQuantumnessEntanglement,adesso2016MeasuresApplicationsQuantum}, is known to have a close relation to quantum coherence even from the early works in quantum optics; the nonclassicality of light was shown to be a source of quantum entanglement \cite{mskim2002EntanglementBeamSplitter,xbwang2002PropertiesBeamsplitterEntangler}; the relation between nonclassicality of lights and entanglement is further established \cite{asboth2005ComputableMeasureNonclassicality,vogel2014UnifiedQuantificationNonclassicality,hertz2020RelatingEntanglementOptical}. For finite-dimensional systems, quantitative relations between quantum coherence and quantum correlations were established \cite{aberg2006QuantifyingSuperposition,streltsov2011LinkingQuantumDiscord,jma2016ConvertingCoherenceQuantum,hzhu2017OperationalMappingCoherenceEntanglement,hzhu2017CoherenceEntanglementMeasures,regula2018ConvertingMultilevelNonclassicality,lhren2021ResourceConversionOperational}. In particular, it was confirmed that the quantum coherence of a quantum state could be converted to quantum entanglement without supplying further quantum coherence \cite{streltsov2015MeasuringCoherenceEntanglement}, which also implied that a quantum entanglement monotone could induce a quantum coherence monotone for quantum states.

Quantum dynamics enter the scene by changing quantum resources either in quantum states or in other quantum dynamics \cite{nielsen2003QuantumDynamicsPhysical,chiribella2008QuantumCircuitArchitecture,chiribella2008TransformingQuantumOperations,gour2019ComparisonQuantumChannels,saxena2020DynamicalResourceTheory,chen2020EntanglementbreakingSuperchannels,hjkim2021OneShotManipulationEntanglement,regula2021OneShotManipulationDynamical,xyuan2020OneshotDynamicalResource,kfang2020NogoTheoremsQuantum,regula2021FundamentalLimitationsDistillation,puchala2021DephasingSuperchannels}. The intimate relation between quantum coherence and quantum entanglement continues to hold for quantum dynamics: specifically, it was shown that a quantum channel's quantum coherence generating power converts to the quantum entanglement generating power without additional quantum coherence in the process \cite{theurer2020QuantifyingDynamicalCoherence}. In fact, quantum channels have various aspects concerning quantum resources other than resource generating powers; a quantum channel can increase, decrease, erase, or preserve the quantum resources of a quantum state \cite{zanardi2000EntanglingPowerQuantum,horodecki2003EntanglementBreakingChannels,mani2015CoheringDecoheringPower,zanardi2017CoherencegeneratingPowerQuantum,zwliu2017ResourceDestroyingMaps,ben_dana2017ResourceCoherenceStates,kbu2017CoheringPowerQuantum,diaz2018UsingReusingCoherence,rosset2018ResourceTheoryQuantum,zwliu2019ResourceTheoriesQuantum,bauml2019ResourceTheoryEntanglement,lli2020QuantifyingResourceContent,cyhsieh2020ResourcePreservability,xyuan2021UniversalOperationalBenchmarking,takahashi2021CreatingDestroyingCoherence}. Does quantum coherence of a quantum channel convert to quantum entanglement in all such aspects as in the case of quantum states?

To shed light on this problem, we focus on quantum measurements that send quantum states to classical states as quantum channels. The classical output of quantum measurements implies that they can generate neither quantum coherence nor quantum entanglement. However, it is known that entangled quantum measurements are useful to certify quantum resources \cite{branciard2010CharacterizingNonlocalCorrelations,bowles2018DeviceIndependentEntanglementCertification,renou2018SelfTestingEntangledMeasurements,sekatski2018CertifyingBuildingBlocks,bancal2018NoiseResistantDeviceIndependentCertification}. This paper investigates quantum coherence and quantum entanglement of quantum measurements using resource theory framework. We find that, despite the classical outputs, quantum resources of quantum measurements can be formulated without relying upon resources of quantum states, and yet they share analogous intimate relations. Understanding the quantum resources of quantum dynamics would enable us to design more effective algorithms and efficient quantum dynamics for the implementation of a quantum computer in the NISQ era \cite{preskill2018QuantumComputingNISQ}.

\section{Resource theory of quantum measurements}
We briefly review quantum measurements and their transformations, and the resource theory of them with respect to the quantum coherence and the quantum entanglement.

A quantum measurement $ \mathcal{M}_{A} $ on a system $ A $ with $ n $ outcomes is often described by a positive operator valued measure (POVM) $ \mathcal{M}_{A}=\{M_{x}\ge 0 : \sum_{x=0}^{n-1}M_{x}=I_{A}, x=0,\dots,n-1 \} $, which, by Born's rule, determines the outcome statistics of an input state $ \rho_{A} $ as $ \{ p_{x} = \tr_{A}\rho_{A}M_{x}: x=0,\dots,n-1 \}  $. The quantum measurement $ \mathcal{M}_{A} $ is equivalently described as a quantum-classical channel that sends a quantum state to a classical state as
\begin{equation}
	\mathcal{M}_{A}(X_{A}) = \sum_{x=1}^{n} \tr\left (M_{x}X_{A} \right ) \ketbra{x}{x}_{R},
\end{equation}
where the system $ R $ is a classical register system \cite{watrous2018TTQI}; we use the same calligraphic letter $ \mathcal{M}_{A} $ both for the POVM and for the above measurement channel with a slight abuse of notation. The convex set of quantum measurements on $ d $ dimensional systems with $ n $ outcomes is denoted by $ \mathbf{M}(d,n) $ \cite{dariano2005ClassicalRandomnessQuantum}; in the following, any single system is assumed to be $ d $ dimensional, and quantum measurements on each system are assumed to have $ n $ outcomes for simplicity.

\begin{figure}[htbp]
    \centering
    \includegraphics{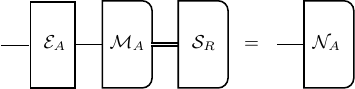}
    \caption{Transformation of a quantum measurement $ \mathcal{M}_{A} $ to a quantum measurement $ \mathcal{N}_{A}$ through a pre-processing channel $ \mathcal{E}_{A} $ and a classical post-processing channel $ \mathcal{S}_{R} $. The double line means classical data.}
	\label{fig: pre-/post-processing of a quantum measurement}
\end{figure}
A quantum measurement $ \mathcal{M}_{A} $ can be converted to another quantum measurement by a pre-processing channel $ \mathcal{E}_{A} $ and a classical post-processing channel $ \mathcal{S}_{R} $ as shown in Fig.~\ref{fig: pre-/post-processing of a quantum measurement} \cite{buscemi2005CleanPOVM,oszmaniec2017SimulatingPOVM}.  A classical post-processing on the outcome effectively results in statistical mixing among the POVM elements of the quantum measurement \cite{buscemi2005CleanPOVM}: consider a classical post-processing channel $ \mathcal{S}_{R} $ that sends an outcome $ x $ to an outcome $ y $ with a probability $ p(y\vert x) $, where $ \sum_{y}p(y\vert x) = 1 $ for all $ x $. It transforms a quantum measurement $ \mathcal{M}_{A} = \{M_{x}\}_{x=0}^{n-1} $ as follows:
\begin{align}
	\mathcal{S}_{R}\circ \mathcal{M}_{A}(\rho_{A}) &= \sum_{x}\tr_{A}(M_{x}\rho_{A}) \sum_{y} p(y\vert x) \ket{y}\! \bra{y}_{R}\\
	&=\sum_{y}\tr_{A}\left (M_{y}'\rho_{A}\right )  \ket{y}\! \bra{y}_{R},
\end{align}
where $ \mathcal{M}_{A}'=\{M_{y}'=\sum_{x}p(y\vert x)M_{x}\}_{y=0}^{n-1} $ is a valid quantum measurement satisfying that $ M_{y}' \ge 0$ and $ \sum_{y}M_{y}'=I_{A} $.

A pre-processing channel $ \mathcal{E}_{A} $ for a quantum measurement $ \mathcal{M}_{A} $ can also be described by its action on the POVM elements considering the output statistics as follows:
\begin{equation}\label{key}
	p_{x}=\tr_{A} \left[ M_{x}\mathcal{E}_{A}(\rho_{A}) \right] = \tr_{A} \left[ \mathcal{E}_{A}^{\dag}\left( M_{x} \right) \rho_{A} \right],
\end{equation}
where $ \mathcal{E}_{A}^{\dag} $ is the adjoint map of $ \mathcal{E}_{A} $ \cite{wilde2017QIT}.
Therefore, a quantum measurement $ \mathcal{M}_{A} $ with a pre-processing channel $ \mathcal{E}_{A} $ is the same as a quantum measurement $ \widetilde{\mathcal{M}}_{A} $:
\begin{equation}\label{key}
	\widetilde{\mathcal{M}}_{A} \equiv \mathcal{M}_{A}\circ \mathcal{E}_{A}=\{\mathcal{E}_{A}^{\dag}(M_{x})\}.
\end{equation}

In the resource theory of quantum coherence, one quantifies quantum coherence with respect to a chosen basis $ \{\ket{i}\} $ the so-called incoherent basis. A quantum state and an operator are incoherent if they are diagonal in the incoherent basis. A quantum measurement $ \mathcal{M}_{A}=\{M_{x}\}_{x=0}^{n-1} $ is called incoherent if all its POVM elements are incoherent, i.e., $ \Delta_{A}M_{x}=M_{x} $ for all $ x $ where $ \Delta_{A} $ is the dephasing channel in the incoherent basis \cite{oszmaniec2019OperationalRelevanceResource,theurer2019QuantifyingOperationsApplication}. We denote the set of the $ n $ outcome incoherent measurements on $ d $ dimensional systems as $ \mathbf{I}(d,n) $.
We take the set of the incoherent measurements as the free resource for quantum coherence of quantum measurements. Operationally, an incoherent measurement $ \mathcal{M}_{A} $ on an input state $ \rho_{A} $ results in an output statistics independent of the quantum coherence of the state as
\begin{align}\label{key}
	p_{x}&=\tr_{A}(\rho_{A}M_{x})\\ &=\tr_{A}(\rho_{A}\Delta(M_{x}))\\
	&=\tr_{A}(\Delta(\rho_{A})M_{x}).
\end{align}
That is, the output statistics depends only on the incoherent part of the input state \cite{theurer2019QuantifyingOperationsApplication}.

For quantum entanglement, a quantum measurement with all its POVM elements being separable operators is called separable; the set of separable measurements is strictly larger than the set of LOCC measurements \cite{bennett1999QuantumNonlocalityEntanglement,childs2013FrameworkBoundingNonlocality}. We take the set of the separable measurements as a free resource \cite{oszmaniec2019OperationalRelevanceResource}; the set of separable measurements on $ d $ dimensional systems $ A $ and $ B $ is denoted as $ \sepm(A\! :\! B) $. Note that entanglement theory does not have any resource destroying channel which destroys entanglement while preserving any separable state \cite{gour2017QuantumResourceTheories,zwliu2017ResourceDestroyingMaps}, analogous to the dephasing channel in the resource theory of quantum coherence. This disallows the operational interpretation of the separable measurements by its outcome statistics' dependence on the entanglement of input states, distinct from the case of the incoherent measurements. However, when the separable measurement is regarded as free, one can still quantify quantum entanglement necessary to implement bipartite measurements which are not separable measurements; such a measure is shown to have operational meanings such as an advantage in the distributed state discrimination \cite{lipka-bartosik2021OperationalSignificanceQuantum}.

Next we ask for the set of free transformations for quantum resources. Firstly, one can easily check that an incoherent measurement stays incoherent under a statistical mixing by a classical post-processing channel; the same holds for the separable measurements. For pre-processing channels, note that the output register system $ R $ of any measurement is treated as being classical; thus we take the register states $ \{\ket{x}_{R}\}_{x=0}^{n-1} $ as the incoherent basis of the system $ R $. Then we figure out the free pre-processing channels for quantum coherence of quantum measurements as follows \cite{theurer2019QuantifyingOperationsApplication}: 
\begin{proposition}
	The set of pre-processing quantum channels that preserves incoherent measurements is the set of detection-incoherent channels $ \mathcal{E}_{A} $ which is characterized by
	\begin{equation}\label{key}
		\Delta_{A} \circ \mathcal{E}_{A} = \Delta_{A}\circ \mathcal{E}_{A}\circ \Delta_{A}.
	\end{equation}
\end{proposition}
For readability, we defer all proofs to the Appendices hereafter.

\section{Resource monotones}
Quantum resources of a quantum channel can be measured by various resource monotones regarding quantum resources in quantum states \cite{gour2019HowQuantifyDynamical}. For quantum measurements, from the definitions of the incoherent measurements and the separable measurements, it is clear that the quantum resources of the POVM elements are essential to the quantum resources in quantum measurements. So we conceive a different relative entropy type quantity between two quantum measurements that aims to measure the quantum resources of the POVM elements. We define the measurement relative entropy between quantum measurements $ \mathcal{M}_{A}=\{M_{x}\}_{x} $ and $ N_{A}=\{N_{x}\}_{x} $ as follows:
\begin{align}\label{key}
	D_{m}(\mathcal{M}_{A}\Vert \mathcal{N}_{A}) &\coloneqq \dfrac{1}{d}D \left (\oplus_{x} M_{x}\Vert \oplus_{x} N_{x} \right )\\
	&=\dfrac{1}{d}\sum_{x} D(M_{x}\Vert N_{x}),
\end{align}
where $ D(\cdot\Vert \cdot) $ is the quantum relative entropy \cite{umegaki1962ConditionalExpectationOperator,gour2021UniquenessOptimalityDynamical,watrous2018TTQI} defined as
\begin{equation}\label{key}
D(M\Vert N)\coloneqq \begin{cases} \tr \{M (\log  M - \log N)\} & \text{im } M \subseteq \text{im } N\\ \infty & \text{else} \end{cases}
\end{equation}
for positive semidefinite operators $ M $ and $ N $; $ \text{im } M $ is the image of $ M $. We use the logarithm base two.

The measurement relative entropy satisfies the following properties:
\begin{lemma}
	Let $ \mathcal{M}_{A}, \mathcal{N}_{A}, \mathcal{K}_{A}, \mathcal{L}_{A}\in \mathbf{M}(d,n) $ be measurement channels, $ \mathcal{E}_{A} $ a unital quantum channel, and $ \mathcal{U}_{A} $ a unitary channel. Let $ \mathcal{S}_{R} $ be a classical channel that sends $ \ket{x}_{R} $ to $ \ket{y}_{R} $ with a probability $ p(y\vert x) $ that satisfies $ \sum_{y} p(y\vert x)=1$ for all $ x $. Let $ 0\le p\le 1 $. The following holds:
	\begin{enumerate}
		\item $ D_{m}(\mathcal{M}_{A}\Vert \mathcal{N}_{A})\ge 0 $; the equality holds if and only if $ \mathcal{M}_{A}=\mathcal{N}_{A} $,
		\item $ D_{m}(\mathcal{M}_{A}\circ \mathcal{E}_{A}\Vert \mathcal{N}_{A}\circ \mathcal{E}_{A})\le D_{m}(\mathcal{M}_{A}\Vert \mathcal{N}_{A}),$
		\item $	D_{m}(\mathcal{M}_{A}\circ \mathcal{U}_{A}\Vert \mathcal{N}_{A}\circ \mathcal{U}_{A}) = D_{m}(\mathcal{M}_{A}\Vert \mathcal{N}_{A}), $
		\item $ D_{m}(\mathcal{S}_{R}\circ \mathcal{M}_{A}\Vert \mathcal{S}_{R}\circ \mathcal{N}_{A}) \le D_{m} (\mathcal{M}_{A}\Vert \mathcal{N}_{A})$,
		\item $	D_{m}(\mathcal{M}_{A}\otimes \mathcal{N}_{B}\Vert \mathcal{K}_{A}\otimes \mathcal{L}_{B}) = D_{m}(\mathcal{M}_{A}\Vert \mathcal{K}_{A}) +D_{m}(\mathcal{N}_{B}\Vert \mathcal{L}_{B}), $
		\item $	D_{m}(p\mathcal{M}_{A}+(1-p)\mathcal{N}_{A}\Vert p\mathcal{K}_{A}+(1-p)\mathcal{L}_{A}) \le pD_{m}(\mathcal{M}_{A}\Vert \mathcal{K}_{A})+(1-p)D_{m}(\mathcal{N}_{A}\Vert \mathcal{L}_{A}).$
	\end{enumerate}
\end{lemma}

We conceive quantum resource monotones for quantum coherence and quantum entanglement, respectively:
\begin{align}
	C_{m}(\mathcal{M}_{A}) &\coloneqq \min_{\mathcal{F}_{A}\in \mathbf{I}(d,n)} D_{m}(\mathcal{M}_{A}\Vert \mathcal{F}_{A}),\\
	E_{m}(\mathcal{M}_{AB})&\coloneqq \min_{\mathcal{F}_{AB}\in \sepm(A : B)} D_{m}(\mathcal{M}_{AB}\Vert \mathcal{F}_{AB}).
\end{align}
Both $ C_{m} $ and $ E_{m} $ are non-negative and faithful thanks to the property of the measurement relative entropy. The monotonicity of $ C_{m} $ under free transformations can be seen as follows: for any unital detection-incoherent channel $ \mathcal{E}_{A} $ and a classical channel $ \mathcal{S}_{R} $,
\begin{align}
	&C_{m}(\mathcal{S}_{R}\circ \mathcal{M}_{A}\circ \mathcal{E}_{A}) \nonumber \\
	&= \min_{\mathcal{F}_{A}\in \mathbf{I}(d,n)} D_{m} (\mathcal{S}_{R}\circ \mathcal{M}_{A}\circ \mathcal{E}_{A} \Vert \mathcal{F}_{A})\nonumber \\
	&\le \min_{\mathcal{F}_{A}\in \mathbf{I}(d,n)} D_{m} (\mathcal{S}_{R}\circ \mathcal{M}_{A}\circ \mathcal{E}_{A} \Vert \mathcal{S}_{R}\circ \mathcal{F}_{A}\circ \mathcal{E}_{A})\nonumber \\
	&\le \min_{\mathcal{F}_{A}'\in \mathbf{I}(d,n)} D_{m} (\mathcal{M}_{A} \Vert \mathcal{F}_{A}'),
\end{align}
where the first inequality is due to the fact that an incoherent measurement remains incoherent after free transformations, and the second inequality is from the monotonicity of the measurement relative entropy. Furthermore, the quantum coherence monotone can be explicitly calculated:
\begin{proposition}
	The quantum coherence monotone of a quantum measurement $ \mathcal{M}_{A}=\{M_{x}\} $ is given by
	\begin{equation}\label{eq: coherence monotone}
		C_{m}(\mathcal{M}_{A})= \dfrac{1}{d}\sum_{x} \left\{ S(\Delta M_{x})-S(M_{x}) \right\},
	\end{equation}
	where $ S(\cdot) $ is the von Neumann entropy.
\end{proposition}
Thus, if we regard $ S(\Delta M_{x})-S(M_{x}) $ as the quantum coherence of the POVM element $ M_{x} $, the quantum coherence monotone $ C_{m}(\mathcal{M}_{A}) $ amounts to the sum of the quantum coherence of all the POVM elements in $ \mathcal{M}_{A} $.

Because entanglement theory does not possess a resource destroying channel \cite{gour2017QuantumResourceTheories}, the entanglement monotone $ E_{m} $ does not possess an analogous expression as eq.~\eqref{eq: coherence monotone}. However, one can still compute the entanglement monotone for some cases, such as the Bell measurement and the Werner measurement: we defer the results to the Appendices for interested readers.

To summarize, taking the set of incoherent measurements as free resource, we regard unital detection-incoherent pre-processing channels with classical post-processing channels as the free transformations; the quantum coherence and entanglement of quantum measurements are quantified by $C_{m}$ and $E_{m}$, respectively.

\section{Quantum coherence conversion to quantum entanglement}
\begin{figure}
	\centering
	\includegraphics{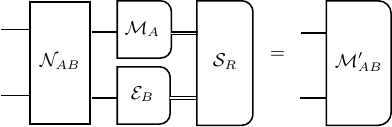}
	\caption{Building a bipartite quantum measurement $ \mathcal{M}_{AB}' $ from two quantum measurements $ \mathcal{M}_{A} $ and $ \mathcal{E}_{B} $ with a pre-processing channel $ \mathcal{N}_{AB} $ and a classical post-processing channel $ \mathcal{S}_{R} $. The double line means classical data.}
	\label{fig: building a bipartite quantum measurement}
\end{figure}
We are now in a position to restate our problem concerning whether quantum coherence of a quantum measurement can be converted into quantum entanglement of a bipartite quantum measurement as depicted in Fig.~\ref{fig: building a bipartite quantum measurement}. Here is our first result:
\begin{theorem}\label{thm: C_m >= E_m}
	Let $ \mathcal{M}_{A} \in \mathbf{M}(d,n)$ be a quantum measurement. For any ancillary incoherent measurement $ \mathcal{E}_{B}\in \mathbf{I}(d,n) $ and a unital detection-incoherent pre-processing channel $ \mathcal{N}_{AB}$, it holds that
	\begin{equation}\label{eq: C_m >= E_m}
		C_{m}(\mathcal{M}_{A})\ge E_{m} (\mathcal{M}_{A}\otimes \mathcal{E}_{B} \circ \mathcal{N}_{AB}).
	\end{equation}
\end{theorem}
This shows that the quantum coherence of a quantum measurement $ \mathcal{M}_{A} $ is an upper bound on the quantum entanglement of any resultant bipartite quantum measurement under the free transformations. Note that a classical post-processing channel is unnecessary in the right-hand side of eq.~\eqref{eq: C_m >= E_m} since it just deteriorates quantum resources as argued before. While it is not always the case that quantum coherence of a quantum measurement fully converts to quantum entanglement, a proper choice of free transformation might achieve the conversion completely as shown in the next result:
\begin{theorem}\label{thm: meas. coh. ent. equality}
	Let $ \mathcal{M}_{A}\in \mathbf{M}(d,n) $ be a quantum measurement. Let $ \mathcal{E}_{B} \in \mathbf{I}(d,n)$ be an incoherent measurement given by
	\begin{equation}\label{key}
		\mathcal{E}_{B}=\begin{cases}
			\{E_{0},\dots,E_{d-1},0,\dots,0 \} & n\ge d,\\
			\{E_{0},\dots,E_{n-2}, I_{B}-\sum_{x=0}^{n-2}E_{x} \} & n < d,
		\end{cases}
	\end{equation}
	where $ E_{x}=\ketbra{x}{x}_{B} $.
	For $ n\ge d $, the following holds:
	\begin{equation}\label{key}
		\sup_{\mathcal{N}_{AB}\in \mathbf{UDI}} E_{m}(\mathcal{M}_{A}\otimes \mathcal{E}_{B}\circ \mathcal{N}_{AB}) = C_{m}(\mathcal{M}_{A}),
	\end{equation}
	where $\mathbf{UDI}$ is the set of unital detection-incoherent channels: an optimal pre-processing channel $ \mathcal{N}_{AB} $ is given by the adjoint channel of the generalized $ \mathrm{CNOT} $ gate.
	For $ n <  d $, the following holds:
	\begin{multline}\label{key}
		\dfrac{n-1}{d}C_{m}(\mathcal{M}_{A}) \le \\
		\sup_{\mathcal{N}_{AB}\in \mathbf{UDI}} E_{m}(\mathcal{M}_{A}\otimes \mathcal{E}_{B}\circ \mathcal{N}_{AB})\\
		 \le C_{m}(\mathcal{M}_{A}).
	\end{multline}
\end{theorem}
When there is a large enough number of measurement outcomes, that is, $ n\ge d $, the quantum coherence completely converts to quantum entanglement for quantum measurement; the class of informationally complete measurements corresponds to this because an informationally complete measurement needs at least $ n\ge d^{2} $ outcomes \cite{busch1991InformationallyCompleteSets,watrous2018TTQI}. In the case of a small number of outcomes $ n<d $, the quantum coherence of a quantum measurement $ \mathcal{M}_{A} $ provides an upper and a lower bound on the quantum entanglement of a bipartite quantum measurement obtained from $ \mathcal{M}_{A} $ without additional coherence: an extreme case of $ n=1 $ corresponds to the trivial measurement $\mathcal{M}_{A}=\{I_{A}\}$ that does not possess quantum coherence.

A typical example of the above result is given by $ \mathcal{M}_{A}=\left\{\ketbra{\pm}{\pm}_{A}:\ket{\pm}=\frac{1}{\sqrt{2}}(\ket{0}_{A}\pm \ket{1}_{A})\right\} $, $ \mathcal{E}_{B}=\{\ketbra{0}{0}_{B}, \ketbra{1}{1}_{B}\} $, and the adjoint channel of the CNOT gate as a pre-processing channel, for which we observe that $ \mathcal{M}_{A}\otimes \mathcal{E}_{B}\circ \mathcal{U}_{\mathrm{CNOT}}^{\dag} = \{\ketbra{\Phi^{\pm}}{\Phi^{\pm}}_{AB},\ketbra{\Psi^{\pm}}{\Psi^{\pm}}_{AB}\} $, where $ \ket{\Phi^{\pm}}_{AB}=\frac{1}{\sqrt{2}}(\ket{00}_{AB}\pm \ket{11}_{AB}) $ and $ \ket{\Psi^{\pm}}_{AB}=\frac{1}{\sqrt{2}}(\ket{01}_{AB}\pm \ket{10}_{AB}) $; The quantum resources are given by $ C_{m}(\mathcal{M}_{A})=E_{m}(\mathcal{M}_{A}\otimes \mathcal{E}_{B}\circ \mathcal{U}_{\mathrm{CNOT}}^{\dag}) =1$.

We emphasize that outputs of any quantum measurements are classical states having no quantum resources; this clearly distinguishes the above results from those on the quantum resource generating powers \cite{theurer2020QuantifyingDynamicalCoherence}.

\section{Coherence monotones from entanglement monotones}
We have seen that the quantum coherence of a quantum measurement can be converted into the quantum entanglement of a bipartite quantum measurement. This implies that, given a quantum entanglement monotone for bipartite quantum measurements, one can utilize it to construct a quantum coherence monotone of a quantum measurement by the convertible amount of the quantum entanglement \footnote{An analogous result for nonclassicality and entanglement in optical modes was known in \cite{asboth2005ComputableMeasureNonclassicality}.}. In the following we show this quantitatively. A quantum coherence monotone is required to satisfy the following properties, that is, non-negativity, faithfulness, monotonicity, and convexity \cite{chitambar2019QuantumResourceTheories,yliu2020OperationalResourceTheory}: for a quantum measurement $\mathcal{M}_{A}$, any unital detection-incoherent channel $ \mathcal{F}_{A} $, and any classical channel $ \mathcal{S}_{R}$,
\begin{enumerate}
	\item $ C(\mathcal{M}_{A}) \ge 0$; $ C(\mathcal{M}_{A}) = 0$ if and only if $ \mathcal{M}_{A}\in \mathbf{I}(d,n) $,
	\item $ C(\mathcal{S}_{R} \circ \mathcal{M}_{A}\circ \mathcal{F}_{A})\le C(\mathcal{M}_{A}) $,
	\item $ C\left (\sum_{i}p_{i}\mathcal{M}_{A}^{(i)}\right )\le \sum_{i}p_{i} C\left (\mathcal{M}_{A}^{(i)}\right ) $,
\end{enumerate}
where $p_{i}\ge 0$, $\sum_{i} p_{i}=1$, and $\mathcal{M}_{A}^{(i)}$'s are quantum measurements.
Similarly a quantum entanglement monotone $ E $ is required to satisfy the following conditions as well: for a quantum measurement $\mathcal{M}_{AB}$, any pre-processing channel $ \mathcal{F}_{AB} $ that preserves $ \sepm(A\! : \! B) $, and any classical channel $ \mathcal{S}_{R}$ acting on the system $ A $ and $ B $,
\begin{enumerate}
	\item $ E(\mathcal{M}_{AB}) \ge 0$; $ E(\mathcal{M}_{AB}) = 0$ if and only if $ \mathcal{M}_{AB}\in \sepm(A\! : \! B) $,
	\item $ E(\mathcal{S}_{R} \circ \mathcal{M}_{AB}\circ \mathcal{F}_{AB})\le E(\mathcal{M}_{AB}) $,
	\item $ E\left (\sum_{i}p_{i}\mathcal{M}_{AB}^{(i)}\right )\le \sum_{i}p_{i} E\left (\mathcal{M}_{AB}^{(i)}\right ) $,
\end{enumerate}
where $p_{i}\ge 0$, $\sum_{i} p_{i}=1$, and $\mathcal{M}_{AB}^{(i)}$'s are quantum measurements.

We figure out that once a quantum entanglement monotone for quantum measurements is given, one can construct a quantum coherence monotone as follows:
\begin{theorem}\label{thm: coh. monotone via ent. monotone}
	Let $ \mathcal{M}_{A}\in \mathbf{M}(d,n) $ be a quantum measurement with $ n > 1 $. Let $ \mathcal{E}_{B} \in \mathbf{I}(d,n)$ be an incoherent measurement given by
	\begin{equation}\label{key}
		\mathcal{E}_{B}=\begin{cases}
			\{E_{0},\dots,E_{d-1},0,\dots,0 \} & n\ge d,\\
			\{E_{0},\dots,E_{n-2}, I_{B}-\sum_{x=0}^{n-2}E_{x} \} & n < d,
		\end{cases}
	\end{equation}
	where $ E_{x}=\ketbra{x}{x}_{B} $.
	A quantum entanglement monotone $ E $ for a quantum measurement induces a quantum coherence monotone for a quantum measurement as follows:
	\begin{equation}\label{key}
		C(\mathcal{M}_{A}) \coloneqq \sup_{\mathcal{F}_{AB}\in \mathbf{UDI}} E(\mathcal{M}_{A}\otimes \mathcal{E}_{B}\circ \mathcal{F}_{AB}),
	\end{equation}
	where $ \mathbf{UDI} $ is the set of unital detection-incoherent channels.
\end{theorem}
This shows that the idea to measure quantum coherence or nonclassicality of a quantum state by its potential to transform to quantum entanglement still holds for the case of quantum measurements \cite{asboth2005ComputableMeasureNonclassicality,streltsov2015MeasuringCoherenceEntanglement,hzhu2017OperationalMappingCoherenceEntanglement}.

\section{Conclusion}
The quantum coherence of a quantum measurement can be converted to the quantum entanglement of a bipartite quantum measurement without additional quantum coherence. We establish this by taking the set of the incoherent measurements as free resources. The set of unital detection-incoherent pre-processing channels with the classical post-processing channels consists of the free transformations for the quantum coherence of quantum measurements. We take the set of the separable measurements as the free resources for entanglement. These quantum resources are measured by resource monotones built upon the measurement relative entropy that we introduce: the measurement relative entropy between two quantum measurements is a sum of the relative entropy between the POVM elements of the quantum measurements so that it helps to capture the quantum resources in each POVM element. Thus, under the free transformations, a quantum measurement could transform to a bipartite quantum measurement of which quantum entanglement is upper bounded by the quantum coherence of the input quantum measurement; quantum coherence of a quantum measurement completely converts to quantum entanglement of a bipartite quantum measurement under the adjoint channel of the generalized CNOT gate as the pre-processing channel.

We also show that the above fact indicates that a quantum entanglement monotone of a quantum measurement induces a quantum coherence monotone of quantum measurements.

Our results strengthen the close relation between quantum coherence and quantum entanglement at the level of quantum dynamics. In the previous work \cite{theurer2020QuantifyingDynamicalCoherence}, it was unavoidable to use the dephasing channel as a pre-processing channel to pinpoint quantum resource generating powers. However, quantum measurements do not generate any quantum resource as outputs; thus, our results enlarge our understanding further in yet another aspect of quantum dynamics. Furthermore, our resource monotones only depend on the quantum measurement without any reference to quantum states distinct from typical dynamical resource monotones \cite{gour2019HowQuantifyDynamical}. Meanwhile, it is desirable to find operational meanings of the measurement relative entropy and resource monotones built on it.

We hope that our research sheds light on the properties of quantum resources of quantum dynamics; the more profound the understanding is, the more effective we can utilize the quantum resources in quantum dynamics for quantum information tasks such as quantum computation in the NISQ era.

\begin{acknowledgements}
H.-J. Kim thanks T. Theurer for his meticulous reading of the paper and advice. 
This research was supported by a National Research Foundation of Korea (NRF) grant funded by the Ministry of Science and ICT (MSIT) (Grants No. NRF-2019R1A2C1006337 and No. NRF-2020M3E4A1079678). S.L. acknowledges support from the Quantum Information Science and Technologies program of the NRF funded by the MSIT (Grant No. 2020M3H3A1105796).
\end{acknowledgements}

\appendix

\section{Resource theory of quantum measurements}
We assume that the outcome register of a quantum measurement channel is a classical system, so we take the measurement outcome basis $ \{\ket{x}_{R}\} $ as the incoherent basis of the register system $ R $. Upon this assumption the set of pre-processing channels that keeps incoherent measurements is given by the detection-incoherent channels \cite{theurer2019QuantifyingOperationsApplication}: 
\begin{proposition}
	The set of pre-processing quantum channels that keeps incoherent measurements is the set of detection-incoherent channels $ \mathcal{E} $ which is characterized by
	\begin{equation}\label{key}
		\Delta \circ \mathcal{E} = \Delta\circ \mathcal{E}\circ \Delta.
	\end{equation}
\end{proposition}
\begin{proof}
	A quantum channel $ \mathcal{E}_{A} $ is detection-incoherent if
	\begin{equation}\label{key}
		\mathcal{E}_{A}^{\dag}\circ \Delta = \Delta \circ \mathcal{E}_{A}^{\dag}\circ \Delta,
	\end{equation}
	where $ \mathcal{E}_{A}^{\dag} $ is the adjoint map of $ \mathcal{E}_{A} $.
	Assume that a pre-processing channel $ \mathcal{E}_{A} $ keeps incoherent POVM elements incoherent such that, for $ M_{x}=\Delta M_{x} $, it holds that $ \mathcal{E}_{A}^{\dag}(M_{x}) = \Delta \circ \mathcal{E}_{A}^{\dag}(M_{x}) $. Then for an arbitrary POVM element $ N_{x} $, it follows that
	\begin{align}
		\mathcal{E}_{A}^{\dag}\circ \Delta (N_{x})&= \mathcal{E}_{A}^{\dag} (\Delta N_{x})\\
		&= \Delta\circ \mathcal{E}_{A}^{\dag}(\Delta N_{x})\\
		&=\Delta \circ \mathcal{E}_{A}^{\dag}\circ \Delta (N_{x}).
	\end{align}
	Thus, we conclude that the set of pre-processing channel that keeps incoherent measurements is the set of detection-incoherent channels.
\end{proof}

Being regarded as a quantum channel, an incoherent measurement channel also belongs to a more stringent class of channels that do not even allow preserving quantum coherence:
\begin{proposition}
	A measurement channel $ \mathcal{M}\in \mathbf{M}(d,n) $ is a classical channel characterized by $ \Delta\circ \mathcal{M}\circ\Delta = \mathcal{M}  $ if and only if it is an incoherent measurement, i.e., $ \Delta(M_{x}) = M_{x} $ for all $ x $.
\end{proposition}
\begin{proof}
	The outcome register of a quantum measurement channel is a classical system so that we have that $ \Delta\circ \mathcal{M} = \mathcal{M} $ for any measurement channel $ \mathcal{M}\in \mathbf{M}(d,n) $. If $ \mathcal{M} $ is a classical channel, that is, $ \mathcal{M}= \Delta\circ \mathcal{M}\circ \Delta $, it follows that
	\begin{align}\label{key}
		\mathcal{M} &=  \Delta\circ \mathcal{M}\circ \Delta\\
		&= \mathcal{M}\circ \Delta\\
		&= \sum_{x} \tr (M_{x}\Delta(\cdot)) \ketbra{x}{x}_{R}\\
		&=\sum_{x} \tr (\Delta(M_{x})\cdot) \ketbra{x}{x}_{R}.
	\end{align}
	Thus we have that $ \mathcal{M}=\{M_{x}\}=\{\Delta (M_{x})\} $.  Conversely, if $ \mathcal{M} $ is an incoherent measurement channel, then tracing back the above equations proves the statement. This completes the proof.
\end{proof}
In addition, a measurement channel is a maximally incoherent operation by definition: hence any measurement channel does not generate coherence.

\begin{proposition}
	All the effects of a bipartite incoherent measurement are separable operators.
\end{proposition}
\begin{proof}
	A POVM element $ M_{xy} $ of a bipartite incoherent measurement satisfies
	\begin{equation}\label{key}
		\Delta_{AB}(M_{AB})= \sum_{x',y'} \bra{x',y'}M_{xy}\ket{x',y'}_{AB} \ketbra{x'}{x'}_{A}\otimes \ketbra{y'}{y'}_{B},
	\end{equation}
thus being a separable operator.
\end{proof}

\subsection{Measurement relative entropy and resource monotones}
We utilize the quantum relative entropy between measurements to construct measurement resource monotones. For $ \mathcal{M}=\{M_{x}\}\in \mathbf{M}(d,n) $ and $ \mathcal{N}=\{N_{x}\}\in \mathbf{M}(d,n) $, we define the measurement relative entropy as
\begin{align}\label{key}
	D_{m}(\mathcal{M}\Vert \mathcal{N}) &\coloneqq \dfrac{1}{d}D \left (\oplus_{x} M_{x}\Vert \oplus_{x} N_{x} \right )\\
	&=\dfrac{1}{d}\sum_{x} D(M_{x}\Vert N_{x}),
\end{align}
where, for $ M\ge 0 $ and $ N\ge 0 $,
\begin{equation}\label{key}
D(M\Vert N)\coloneqq \begin{cases} \tr \{M (\log  M - \log N)\} & \text{if }\text{im } M \subseteq \text{im } N\\ \infty & \text{else} \end{cases}
\end{equation}
is the quantum relative entropy between positive semidefinite operators and $ \text{im } M $ is the image of an operator $ M $ \cite{watrous2018TTQI}.

The measurement relative entropy satisfies the following properties:
\begin{lemma}
	Let $ \mathcal{M}, \mathcal{N}, \mathcal{K}, \mathcal{L}\in \mathbf{M}(d,n) $ be measurement channels, $ \mathcal{E} $ a unital quantum channel, and $ \mathcal{U} $ a unitary channel. Let $ \mathcal{S}_{R} $ be a classical channel that sends $ \ket{x}_{R} $ to $ \ket{y}_{R} $ with a probability $ p(y\vert x) $ that satisfies $ \sum_{y}p(y\vert x)=1 $ for all $ x $. Let $ 0\le p\le 1 $. The following holds:
	\begin{enumerate}
		\item $ D_{m}(\mathcal{M}_{A}\Vert \mathcal{N}_{A})\ge 0 $; the equality holds if and only if $ \mathcal{M}_{A}=\mathcal{N}_{A} $,
		\item $ D_{m}(\mathcal{M}_{A}\circ \mathcal{E}_{A}\Vert \mathcal{N}_{A}\circ \mathcal{E}_{A})\le D_{m}(\mathcal{M}_{A}\Vert \mathcal{N}_{A}),$
		\item $	D_{m}(\mathcal{M}_{A}\circ \mathcal{U}_{A}\Vert \mathcal{N}_{A}\circ \mathcal{U}_{A}) = D_{m}(\mathcal{M}_{A}\Vert \mathcal{N}_{A}), $
		\item $ D_{m}(\mathcal{S}_{R}\circ \mathcal{M}_{A}\Vert \mathcal{S}_{R}\circ \mathcal{N}_{A}) \le D_{m} (\mathcal{M}_{A}\Vert \mathcal{N}_{A})$,
		\item $	D_{m}(\mathcal{M}_{A}\otimes \mathcal{N}_{B}\Vert \mathcal{K}_{A}\otimes \mathcal{L}_{B}) = D_{m}(\mathcal{M}_{A}\Vert \mathcal{K}_{A}) +D_{m}(\mathcal{N}_{B}\Vert \mathcal{L}_{B}), $
		\item $	D_{m}(p\mathcal{M}_{A}+(1-p)\mathcal{N}_{A}\Vert p\mathcal{K}_{A}+(1-p)\mathcal{L}_{A}) \le pD_{m}(\mathcal{M}_{A}\Vert \mathcal{K}_{A})+(1-p)D_{m}(\mathcal{N}_{A}\Vert \mathcal{L}_{A}).$
	\end{enumerate}
\end{lemma}
\begin{proof}
	\begin{enumerate}
		\item The non-negativity and the faithfulness of the measurement relative entropy follow from the properties of the quantum relative entropy.
		\item The measurement relative entropy is monotone under any unital pre-processing channel $ \mathcal{E} $:
	\begin{align}\label{key}
		&D_{m}(\mathcal{M}_{A}\circ \mathcal{E}_{A}\Vert \mathcal{N}_{A}\circ \mathcal{E}_{A}) \nonumber\\
		&= \dfrac{1}{d}\sum_{x} D(\mathcal{E}_{A}^{\dag} (M_{x})\Vert \mathcal{E}_{A}^{\dag}(N_{x}))\\
		&\le D_{m}(\mathcal{M}_{A}\Vert \mathcal{N}_{A}),
	\end{align}
	where we interpreted the action of the pre-processing channel $ \mathcal{E} $ through its adjoint channel on the POVM elements regarding the measurement outcome probabilities. Since $ \mathcal{E} $ is a unital quantum channel, its adjoint map $ \mathcal{E}^{\dag} $ is also a unital quantum channel. So the inequality follows from the monotonicity of the quantum relative entropy.
	\item The measurement relative entropy is invariant under any unitary pre-processing channel $ \mathcal{U} $ due to the invariance of the quantum relative entropy under isometries.
	\item The measurement relative entropy is monotone decreasing under a classical post-processing channel:
	\begin{align}
		&D_{m}(\mathcal{S}_{R}\circ \mathcal{M}_{A}\Vert \mathcal{S}_{R}\circ \mathcal{N}_{A})\nonumber\\&= \dfrac{1}{d}\sum_{y}D\left (\sum_{x}p(y\vert x) M_{x}\Vert \sum_{x}p(y\vert x)N_{x}\right )\\
		&\le \dfrac{1}{d} \sum_{y}\sum_{x} D(p(y\vert x) M_{x}\Vert p(y\vert x) N_{x})\\
		&= \dfrac{1}{d} \sum_{y}\sum_{x}p(y\vert x)D(M_{x}\Vert N_{x})\\
		&=\dfrac{1}{d}\sum_{x} D(M_{x}\Vert N_{x})\\
		&=D_{m}(\mathcal{M}_{A}\Vert \mathcal{N}_{A}),
	\end{align}
	where the first inequality and the third line follow from
	\begin{gather}\label{key}
		D(P_{0}+P_{1}\Vert Q_{0}+Q_{1})\le D(P_{0}\Vert Q_{0})+D(P_{1}\Vert Q_{1}),\\
		D(\alpha P_{0}\Vert \beta Q_{0})=\alpha D(P_{0}\Vert Q_{0}) +(\alpha \log \alpha/\beta) \tr P_{0}
	\end{gather}
	for any positive semidefinite operators $ P_{0} $, $P_{1} $, $ Q_{0} $, and $ Q_{1}$, and $ \alpha, \beta > 0 $; the fourth line comes from $ \sum_{y}p(y\vert x) = 1 $ for all $ x $.
	
	\item The measurement relative entropy is additive for the tensor product:
	\begin{align}\label{key}
		& D_{m}(\mathcal{M}_{A}\otimes \mathcal{N}_{B}\Vert \mathcal{K}_{A}\otimes \mathcal{L}_{B})\nonumber\\
		&= \dfrac{1}{d^{2}} \sum_{x,y} D(M_{x}\otimes N_{y}\Vert K_{x}\otimes L_{y})\\
		&=\dfrac{1}{d^{2}} \sum_{x,y} \{ ( \tr_{B} N_{y} )D(M_{x}\Vert K_{x})\nonumber\\
		&\quad+( \tr_{A}M_{x} ) D(N_{y}\Vert L_{y}) \}\\
		&=\dfrac{1}{d} \sum_{x} D(M_{x}\Vert K_{x})+ \dfrac{1}{d} \sum_{y}D(N_{y}\Vert L_{y}) \\
		 &= D_{m}(\mathcal{M}_{A}\Vert \mathcal{K}_{A}) +D_{m}(\mathcal{N}_{B}\Vert \mathcal{L}_{B}).
	\end{align}
	
	\item The measurement relative entropy is jointly convex due to the joint convexity of the quantum relative entropy:
	\begin{align}
		& D_{m}(p\mathcal{M}_{A}+(1-p)\mathcal{N}_{A}\Vert p\mathcal{K}_{A}+(1-p)\mathcal{L}_{A}) \\
		&= \dfrac{1}{d}\sum_{x}D(p M_{x}+(1-p) N_{x}\Vert pK_{x}+(1-p)L_{x} )\\
		&\le \dfrac{1}{d}\sum_{x}\{pD(M_{x}\Vert K_{x} )+(1-p) D(N_{x}\Vert L_{x} )\}\\
		&= pD_{m}(\mathcal{M}_{A}\Vert \mathcal{K}_{A})+(1-p)D_{m}(\mathcal{N}_{A}\Vert \mathcal{L}_{A}).
	\end{align}
	\end{enumerate}
\end{proof}

Now we construct a quantum coherence and  quantum entanglement monotones for  quantum measurement channels using the measurement relative entropy as follows:
\begin{align}
	C_{m}(\mathcal{M}_{A}) &\coloneqq \min_{\mathcal{F}_{A}\in \mathbf{I}(d,n)} D_{m}(\mathcal{M}_{A}\Vert \mathcal{F}_{A}),\\
	E_{m}(\mathcal{M}_{AB})&\coloneqq \min_{\mathcal{F}_{AB}\in \sepm(A : B)} D_{m}(\mathcal{M}_{AB}\Vert \mathcal{F}_{AB}),
\end{align}
where $ \sepm(A\! :\! B) $ is the set of separable measurements.

The above resource monotones are non-negative and faithful since the quantum relative entropy is non-negative and faithful. The same holds for $ E_{m} $ for separable measurements. The quantum coherence monotone $ C_{m} $ is also monotone decreasing under any unital detection-incoherent (UDI) pre-processing channels and the classical post-processing channels: for a UDI channel $ \mathcal{E}_{A} $ and a classical post-processing channel $ \mathcal{S}_{R} $, it follows that
\begin{align}
	&C_{m}(\mathcal{S}_{R}\circ \mathcal{M}_{A}\circ \mathcal{E}_{A}) \nonumber\\
	&= \min_{\mathcal{F}_{A}\in \mathbf{I}(d,n)} D_{m} (\mathcal{S}_{R}\circ \mathcal{M}_{A}\circ \mathcal{E}_{A} \Vert \mathcal{F}_{A})\\
	&\le \min_{\mathcal{F}_{A}\in \mathbf{I}(d,n)} D_{m} (\mathcal{S}_{R}\circ \mathcal{M}_{A}\circ \mathcal{E}_{A} \Vert \mathcal{S}_{R}\circ \mathcal{F}_{A}\circ \mathcal{E}_{A})\\
	&\le \min_{\mathcal{F}_{A}\in \mathbf{I}(d,n)} D_{m} (\mathcal{M}_{A} \Vert \mathcal{F}_{A}),
\end{align}
where we used the monotonicity of $ D_{m} $ in the last inequality.

Note that the  quantum coherence monotone for measurement channels can be explicitly calculated:
\begin{proposition}
	The quantum coherence of a quantum measurement $ \mathcal{M}_{A}=\{M_{x}\} $ is given as follows:
	\begin{align}
		C_{m}(\mathcal{M}_{A})&=\dfrac{1}{d}\sum_{x} D(M_{x}\Vert \Delta M_{x})\\
		&= \dfrac{1}{d}\sum_{x} \left\{ S(\Delta M_{x})-S(M_{x}) \right\}\\
		&= \dfrac{1}{d}\sum_{x} p_{x} C_{r}(\rho_{x}),
	\end{align}
	where $ S(\cdot) $ is the von Neumann entropy, $ C_{r}(\rho) $ is the relative entropy of coherence for quantum states, and $ \rho_{x} \equiv M_{x}/\tr M_{x} $ for all $ x $.
\end{proposition}
\begin{proof}
	Let $ M_{x} = p_{x}\rho_{x} $ with $ p_{x}=\tr M_{x} $. During the derivation, we also denote $ F_{x} = q_{x} \sigma_{x} $ with $ q_{x} = \tr F_{x} $:
	\begin{align}
		C_{m}(\mathcal{M}_{A})&= \min_{\mathcal{F}_{A}\in \mathbf{I}(d,n)} D_{m}(\mathcal{M}_{A}\Vert \mathcal{F}_{A})\\
		&= \min_{\mathcal{F}_{A}\in \mathbf{I}(d,n)} \dfrac{1}{d}\sum_{x} D(M_{x}\Vert F_{x})\\
		&= \min_{\mathcal{F}_{A}\in \mathbf{I}(d,n)} \dfrac{1}{d}\sum_{x} D(p_{x}\rho_{x}\Vert q_{x} \sigma_{x})\\
		&= \min_{\mathcal{F}_{A}\in \mathbf{I}(d,n)} \dfrac{1}{d}\sum_{x} \left\{ p_{x}D(\rho_{x}\Vert  \sigma_{x}) +p_{x}\log \dfrac{p_{x}}{q_{x}}\right\}\\
		&= \min_{\mathcal{F}_{A}\in \mathbf{I}(d,n)} \dfrac{1}{d}\left\{ \sum_{x} p_{x}D(\rho_{x}\Vert  \sigma_{x}) + D(\vec{p}\Vert \vec{q}) \right\}.
	\end{align}
	The last line implies that the minimization is achieved by incoherent measurements $ \mathcal{F}_{A} $ such that $ \tr F_{x} =\tr M_{x}$, that is, $ q_{x}=p_{x} $ for all $ x $ due to the non-negativity of the quantum relative entropy. Applying this fact, we conclude that
	\begin{align}
		C_{m}(\mathcal{M}_{A})&= \dfrac{1}{d}\sum_{x}p_{x} D(\rho_{x}\Vert \Delta \rho_{x})\\
		&=\dfrac{1}{d}\sum_{x}D(M_{x}\Vert \Delta M_{x})\\
		&= \dfrac{1}{d}\sum_{x} \left\{ S( \Delta M_{x}) - S(M_{x}) \right\}.
	\end{align}
\end{proof}

As some examples of quantum measurements regarding quantum resources, a quantum measurement $ \mathcal{M}_{A}=\{ \ketbra{\pm}{\pm}_{A}: \ket{\pm}=\frac{1}{\sqrt{2}}(\ket{0}\pm \ket{1})\} $ has $ C_{m}(\mathcal{M}_{A})=1 $, while an incoherent measurement $ \mathcal{E}_{A}=\{\ketbra{0}{0}_{A},\ketbra{1}{1}_{A}\} $ has $ C_{m}(\mathcal{E}_{A}) =0 $.

For  quantum entanglement, the Bell measurement $ \mathcal{M}_{AB}=\{\Phi_{AB}^{\pm},\Psi_{AB}^{\pm}\} $ has $ E_{m}(\mathcal{M}_{AB})=1 $ with an optimal free measurement
	\begin{align}
	\mathcal{F}_{AB}= &\bigg\{ \frac{1}{2}(\ketbra{00}{00}_{AB}+\ketbra{11}{11}_{AB}),\nonumber\\
	&\;\;\frac{1}{2}(\ketbra{00}{00}_{AB}+\ketbra{11}{11}_{AB}),\nonumber\\
	&\;\;\frac{1}{2}(\ketbra{01}{01}_{AB}+\ketbra{10}{10}_{AB}),\nonumber\\
	&\;\;\frac{1}{2}(\ketbra{01}{01}_{AB}+\ketbra{10}{10}_{AB})\bigg \}.
	\end{align}
	
As another example, we consider a class of two-qubit Bell-diagonal measurements given by
\begin{multline}
    \mathcal{B}_{AB} = \{ \mathcal{U}_{A}(p_{1}\Phi_{AB}^{+} + p_{2}\Phi_{AB}^{-}+p_{3}\Psi_{AB}^{+}+p_{4}\Psi_{AB}^{-}):\\
    U_{A} \in \{I_{A}, \sigma_{A}^{X}, \sigma_{A}^{Y}, \sigma_{A}^{Z}\}\},
\end{multline}
where $ p_{1}, p_{2}, p_{3}, p_{4} \ge 0$, $\sum_{i=1}^{4} p_{i}=1$, and $ \sigma_{A}^{X}, \sigma_{A}^{Y}, \sigma_{A}^{Z}$ are the Pauli operators. Without loss of generality, we assume that $\max_{i} p_{i}=p_{1}$. Each POVM element is the Bell-diagonal state which is known to be entangled if and only if $ p_{1} > \frac{1}{2}$ \cite{horodecki1996InformationtheoreticAspectsInseparability,horodecki1997inseparable}. For $p_{1}> \frac{1}{2}$, one can compute the entanglement monotone of $\mathcal{B}_{AB}$ utilizing the relative entropy of entanglement for each POVM element \cite{vedral1997} as $E_{m} ( \mathcal{B}_{AB}) = 1 - h(p_{1}) $, where $h(p_{1}) = -p_{1} \log p_{1} - (1-p_{1}) \log (1- p_{1})$ is the binary entropy. An optimal separable measurement is given by
\begin{align}
    \mathcal{F}_{AB} &= \bigg\{ \mathcal{U}_{A}\bigg(\frac{1}{2}\Phi_{AB}^{+} + \dfrac{p_{2}}{2 (1 - p_{1})}\Phi_{AB}^{-}+\dfrac{p_{3}}{2 (1 - p_{1})}\Psi_{AB}^{+}\nonumber\\
    &\;\; +\dfrac{p_{4}}{2 (1 - p_{1})}\Psi_{AB}^{-}\bigg): U_{A} \in \{I_{A}, \sigma_{A}^{X}, \sigma_{A}^{Y}, \sigma_{A}^{Z}\}\bigg\}.
\end{align}
 
 An example of the above class is a two-qubit measurement given by
 \begin{equation}
     \mathcal{W}_{AB} = \left \{ p\Phi_{AB}^{\pm} + \frac{1-p}{4}I_{AB},\; p\Psi_{AB}^{\pm} + \frac{1-p}{4}I_{AB} \right \},
 \end{equation}
 where $ 0\le p\le 1$. The POVM elements of the measurement are equal to the Werner state up to local unitary operations so that each of them is known to be entangled for $ p> \frac{1}{3} $. The entanglement monotone of the measurement for $ p > \frac{1}{3} $ is computed as $ E_{m}(\mathcal{W}_{AB}) = 1 - h(\lambda) $, where $ \lambda = \frac{1+3p}{4} $; $ E_{m}(\mathcal{W}_{AB}) = 0 $ for $p \le \frac{1}{3} $. An optimal free POVM element for $ \mathcal{W}_{AB}$ is given by $ \left\{\frac{1}{3}\Phi_{AB}^{\pm} + \frac{1}{6}I_{AB}, \frac{1}{3}\Psi_{AB}^{\pm} + \frac{1}{6}I_{AB}\right\} $.
 
 Another example of the above class is a two-qubit measurement given by
 \begin{multline}
    \mathcal{I}_{AB} = \bigg\{ \mathcal{U}_{A}\left(p\Phi_{AB}^{+} + \frac{1-p}{3}(I_{AB}-\Phi_{AB}^{+})\right):\\
    U_{A} \in \{I_{A}, \sigma_{A}^{X}, \sigma_{A}^{Y}, \sigma_{A}^{Z}\}\bigg\},
\end{multline}
 where $ 0\le p\le 1$. The POVM elements of the measurement are equal to the isotropic state up to local unitary operations so that each of them is known to be entangled for $ p> \frac{1}{2} $. The entanglement monotone of the measurement for $ p > \frac{1}{2} $ is computed as $ E_{m}(\mathcal{I}_{AB}) = 1 - h( p ) $; $ E_{m}(\mathcal{I}_{AB}) = 0 $ for $p \le \frac{1}{2} $. An optimal free POVM element for $ \mathcal{I}_{AB}$ is given by
\begin{multline}
    \mathcal{F}_{AB} = \bigg\{ \mathcal{U}_{A}\left(\dfrac{1}{2}\Phi_{AB}^{+} + \dfrac{1}{6}(I_{AB}-\Phi_{AB}^{+})\right):\\
    U_{A} \in \{I_{A}, \sigma_{A}^{X}, \sigma_{A}^{Y}, \sigma_{A}^{Z}\}\bigg\}.
\end{multline}.

\section{Quantum coherence conversion to quantum entanglement}
The quantum coherence of a measurement $ \mathcal{M}_{A}$ upper-bounds the quantum entanglement of a composite measurement that is constructed from $ \mathcal{M}_{A}$ using free resources:
\begin{theorem}\label{thm: C_m >= E_m}
	Let $ \mathcal{M}_{A} \in \mathbf{M}(d,n)$ be a quantum measurement. For any ancillary incoherent measurement $ \mathcal{E}_{B}\in \mathbf{I}(d,n) $ and a unital detection-incoherent pre-processing channel $ \mathcal{N}_{AB}$, it holds that
	\begin{equation}\label{key}
		C_{m}(\mathcal{M}_{A})\ge E_{m} (\mathcal{M}_{A}\otimes \mathcal{E}_{B} \circ \mathcal{N}_{AB}) .
	\end{equation}
\end{theorem}
\begin{proof}
	Let an optimal incoherent measurement for $ C_{m}(\mathcal{M}_{A}) $ be $ \mathcal{F}_{A}^{\ast} $. It follows that
	\begin{align}
		C_{m}(\mathcal{M}_{A})&= \min_{\mathcal{F}_{A}\in \mathbf{I}(d,n)} D_{m}(\mathcal{M}_{A}\Vert \mathcal{F}_{A})\\
		&=D_{m}(\mathcal{M}_{A}\Vert \mathcal{F}_{A}^{\ast})\\
		&=D_{m}(\mathcal{M}_{A}\otimes \mathcal{E}_{B}\Vert \mathcal{F}_{A}^{\ast}\otimes \mathcal{E}_{B})\\
		&\ge D_{m}(\mathcal{M}_{A}\otimes \mathcal{E}_{B} \circ \mathcal{N}_{AB}\Vert \mathcal{F}_{A}^{\ast}\otimes \mathcal{E}_{B}\circ \mathcal{N}_{AB})\\
		&\ge \min_{\mathcal{F}_{AB}' \in \sepm(A:B)} D_{m}(\mathcal{M}_{A}\otimes \mathcal{E}_{B} \circ \mathcal{N}_{AB}\Vert \mathcal{F}_{AB}')\\
		&=E_{m} (\mathcal{M}_{A}\otimes \mathcal{E}_{B} \circ \mathcal{N}_{AB}),
	\end{align}
	where we used the fact that $ \mathcal{F}_{A}^{\ast}\otimes \mathcal{E}_{B}\circ \mathcal{N}_{AB}\in \mathbf{I}(d\times d, n\times n)\subset \sepm(A\! :\! B) $ in the last inequality.
\end{proof}
Note that it is unnecessary to consider a classical post-processing channel since it does not increase quantum entanglement.

Before moving into the main result, we extend the relative entropy of entanglement for bipartite states to positive semidefinite bipartite operators, or unnormalized bipartite states in other words. Recall that the von Neumann entropy and the quantum relative entropy are defined over positive semidefinite operators \cite{watrous2018TTQI}:
\begin{multline}\label{key}
	E_{R}(X_{AB}) \coloneqq \min \{D(X_{AB}\Vert Y_{AB}) :\\
	Y_{AB}\in \sep(A:B), \tr_{AB} Y_{AB}= \tr_{AB} X_{AB}\},
\end{multline}
where $ \sep(A\!:\!B) $ denotes the set of separable operators. We first extend some of the results in \cite{plenio2000OperatorMonotonesReduction} to the set of positive semidefinite operators:
\begin{lemma}
	For a positive semidefinite operator $ X_{AB} $ and a separable operator $ Y_{AB} $, it holds that
	\begin{gather}
		S(X_{A}) - S(X_{AB}) \le D(X_{AB}\Vert Y_{AB}) - D(X_{A}\Vert Y_{A}),\\
		S(X_{B}) - S(X_{AB}) \le D(X_{AB}\Vert Y_{AB}) - D(X_{B}\Vert Y_{B}).
	\end{gather}
\end{lemma}
\begin{proof}
	The map $ \Lambda_{B}(Z_{B})=\tr_{B} (Z_{B})I_{B}-Z_{B} $ is positive but not completely positive \cite{horodecki1999reduction}. Since $ Y_{AB} $ is separable, it is undistillable so that it satisfies $ \mathsf{Id}_{A}\otimes \Lambda_{B} (Y_{AB})= Y_{A}\otimes I_{B}- Y_{AB}\ge 0 $, where $ \mathsf{Id}_{A}$ is the identity channel. From this, we have that
	\begin{gather}
		\log Y_{A}\otimes I_{B}\ge \log Y_{AB},\\
		\tr_{AB} X_{AB} \log Y_{A}\otimes I_{B}\ge \tr_{AB} X_{AB}\log Y_{AB},\\
		-S(X_{AB})+S(X_{A})-S(X_{A})-\tr_{AB} X_{AB} \log Y_{A}\otimes I_{B}\nonumber\\
		\le -S(X_{AB})-\tr_{AB} X_{AB}\log Y_{AB},\\
		S(X_{A})-S(X_{AB}) \le D(X_{AB}\Vert Y_{AB}) - D(X_{A}\Vert Y_{A}).
	\end{gather}

	The second one can be derived similarly.
\end{proof}
\begin{lemma}\label{lem: E_R lower-bound by the condition entropy}
	For a positive semidefinite matrix $ X_{AB} $, it holds that
	\begin{equation}\label{key}
		E_{R} (X_{AB})\ge \max\{S(X_{A})-S(X_{AB}), S(X_{B})-S(X_{AB})\}
	\end{equation}
\end{lemma}
\begin{proof}
	Let $ E_{R}(X_{AB}) = D(X_{AB}\Vert Y_{AB}^{\ast}) $. Then
	\begin{align}
		S(X_{A})-S(X_{AB})&\le D(X_{AB}\Vert Y_{AB}^{\ast}) - D(X_{A}\Vert Y_{A}^{\ast})\\
		&\le D(X_{AB}\Vert Y_{AB}^{\ast})\\
		&=E_{R}(X_{AB}).
	\end{align}
	The remaining one can be shown similarly.
\end{proof}

Upon the above lemmata, we obtain the following result:
\begin{lemma}\label{lem: E_m >= C_m with cnot}
	Let $ \mathcal{M}_{A}\in \mathbf{M}(d,n) $ be a quantum measurement and $ \mathcal{U}_{\cnot}=\sum_{i,j}\ket{i,j\oplus i}\!\bra{i,j} $ the generalized $ \mathrm{CNOT} $ gate \footnote{$ \oplus $ is addition mod $ d $.}. Let $ \mathcal{E}_{B} \in \mathbf{I}(d,n)$ be an incoherent measurement given by
	\begin{equation}\label{key}
		\mathcal{E}_{B}=\begin{cases}
			\{E_{0},\dots,E_{d-1},0,\dots,0 \} & n\ge d,\\
			\{E_{0},\dots,E_{n-2}, I_{B}-\sum_{x=0}^{n-2}E_{x} \} & n < d,
		\end{cases}
	\end{equation}
	where $ E_{x}=\ketbra{x}{x}_{B} $. The following holds:
	\begin{equation}\label{key}
		E_{m}(\mathcal{M}_{A}\otimes \mathcal{E}_{B}\circ \mathcal{U}_{\cnot}^{\dag}) \ge \begin{cases}
			C_{m}(\mathcal{M}_{A}) & n\ge d,\\
			\dfrac{n-1}{d}C_{m}(\mathcal{M}_{A}) & n<d.
		\end{cases}
	\end{equation}
\end{lemma}
\begin{proof}
	Note that the composite measurement consisting of $ \mathcal{M}_{A}\in \mathbf{I}(d,n) $ and $ \mathcal{N}_{B}\in \mathbf{I}(d,n) $ is an element of $ \mathbf{I}(d\times d, n\times n) $. $ \mathcal{U}_{\cnot}^{\dag} $ is a unital detection-incoherent channel since its adjoint channel is a maximally incoherent operation. The case of $ n\ge d $ can be proven as follows:
	\begin{align}
		&E_{m}(\mathcal{M}_{A}\otimes \mathcal{E}_{B}\circ \mathcal{U}_{\cnot}^{\dag}) \nonumber\\
		&=\min_{\mathcal{F}_{AB}\in \sepm(A:B)} D_{m}\left (\mathcal{M}_{A}\otimes \mathcal{E}_{B}\circ \mathcal{U}_{\cnot}^{\dag}\Vert \mathcal{F}_{AB}\right )\\
		&=\min_{\mathcal{F}_{AB}\in \sepm(A:B)}\dfrac{1}{d^{2}}D\left( \oplus_{x, y}\mathcal{U}_{\cnot}(M_{x}\otimes E_{y})\Vert \oplus_{x, y} F_{x y} \right)\\
		&=\min_{\mathcal{F}_{AB}\in \sepm(A:B)} \dfrac{1}{d^{2}}\sum_{x, y=0}^{n-1} D\left( \mathcal{U}_{\cnot}(M_{x}\otimes E_{y})\Vert  F_{x y} \right)\\
		&\ge \dfrac{1}{d^{2}}\sum_{x, y=0}^{n-1} E_{R}(\mathcal{U}_{\cnot}(M_{x}\otimes E_{y}))\\
		&= \dfrac{1}{d}\sum_{x=0}^{n-1} E_{R}(\mathcal{U}_{\cnot}(M_{x}\otimes E_{0}))\\
		&\ge \dfrac{1}{d}\sum_{x=0}^{n-1} \left\{ S(\Delta M_{x}) - S(M_{x}) \right\}\\
		&= C_{m}(\mathcal{M}_{A}),
	\end{align}
	where the fifth line follows from the fact that $ E_{R}(\mathcal{U}_{\cnot}(M_{x}\otimes E_{y}))=E_{R}(\mathcal{U}_{\cnot}(M_{x}\otimes E_{0})) $ for all $ y $ because of
	\begin{equation}\label{key}
		\mathcal{U}_{\cnot}(M_{x}\otimes E_{y})=\mathsf{Id}_{A}\otimes \mathcal{S}_{y}\circ\mathcal{U}_{\cnot}(M_{x}\otimes E_{0})
	\end{equation}
	with the (unitary) shift channel $ S_{y}=\sum_{i}\ket{i\oplus y}\!\bra{i} $ (or the generalized Pauli $ X $ channel); the inequality follows from Lemma~\ref{lem: E_R lower-bound by the condition entropy}. For $ n<d $, it can be seen in a similar way:
	\begin{align}
		&E_{m}(\mathcal{M}_{A}\otimes \mathcal{E}_{B}\circ \mathcal{U}_{\cnot}^{\dag}) \nonumber\\
		&=\min_{\mathcal{F}_{AB}\in \sepm(A:B)} D_{m}\left (\mathcal{M}_{A}\otimes \mathcal{E}_{B}\circ \mathcal{U}_{\cnot}^{\dag}\Vert \mathcal{F}_{AB}\right )\\
		&\ge \dfrac{1}{d^{2}}\sum_{x, y=0}^{n-1} E_{R}(\mathcal{U}_{\cnot}(M_{x}\otimes E_{0}))\\
		&\ge \dfrac{n-1}{d^{2}}\sum_{x=0}^{n-1} E_{R}(\mathcal{U}_{\cnot}(M_{x}\otimes E_{0}))\\
		&\ge \dfrac{n-1}{d^{2}}\sum_{x=0}^{n-1} \left\{ S(\Delta M_{x}) - S(M_{x}) \right\}\\
		&= \dfrac{n-1}{d} C_{m}(\mathcal{M}_{A}).
	\end{align}
	 This completes the proof.
\end{proof}

Note that for information complete measurements, it holds that $ n\ge d^{2} $. Upon the above results, we arrive at the main result \footnote{Classical post-processing channels are unnecessary since they do not increase quantum coherence or quantum entanglement.}:
\begin{theorem}\label{thm: meas. coh. ent. equality}
	Let $ \mathcal{M}_{A}\in \mathbf{M}(d,n) $ be a quantum measurement. Let $ \mathcal{E}_{B} \in \mathbf{I}(d,n)$ be an incoherent measurement given by
	\begin{equation}\label{key}
		\mathcal{E}_{B}=\begin{cases}
			\{E_{0},\dots,E_{d-1},0,\dots,0 \} & n\ge d,\\
			\{E_{0},\dots,E_{n-2}, I_{B}-\sum_{x=0}^{n-2}E_{x} \} & n < d,
		\end{cases}
	\end{equation}
	where $ E_{x}=\ketbra{x}{x}_{B} $.
	For $ n\ge d $, the following holds:
	\begin{equation}\label{key}
		\sup_{\mathcal{N}_{AB}\in \mathbf{UDI}} E_{m}(\mathcal{M}_{A}\otimes \mathcal{E}_{B}\circ \mathcal{N}_{AB}) = C_{m}(\mathcal{M}_{A}),
	\end{equation}
	where $ \mathbf{UDI} $ denotes the set of unital detection incoherent channels: an optimal pre-processing channel $ \mathcal{N}_{AB} $ is given by the adjoint channel of the generalized $ \mathrm{CNOT} $ gate.
	For $ n <  d $, the following holds:
        \begin{multline}
			\dfrac{n-1}{d}C_{m}(\mathcal{M}_{A})\le\\
			\sup_{\mathcal{N}_{AB}\in \mathbf{UDI}} E_{m}(\mathcal{M}_{A}\otimes \mathcal{E}_{B}\circ \mathcal{N}_{AB})\\
			\le C_{m}(\mathcal{M}_{A}).
		\end{multline}
\end{theorem}
\begin{proof}
	Theorem \ref{thm: C_m >= E_m} shows that
	\begin{equation}\label{key}
		E_{m}(\mathcal{M}_{A}\otimes \mathcal{E}_{B}\circ \mathcal{N}_{AB})\le C_{m}(\mathcal{M}_{A})
	\end{equation}
	for any unital detection-incoherent channel $ \mathcal{N}_{AB} $. On the other hand, using $ \mathcal{U}_{\cnot}^{\dag} $ as the preprocessing channel, Lemma \ref{lem: E_m >= C_m with cnot} indicates that
	\begin{equation}\label{key}
		E_{m}(\mathcal{M}_{A}\otimes \mathcal{E}_{B}\circ \mathcal{U}_{\cnot}^{\dag})\begin{cases}
			\ge C_{m}(\mathcal{M}_{A}) & n\ge d,\\
			\ge \dfrac{n-1}{d}C_{m}(\mathcal{M}_{A}) & n< d.
		\end{cases}
	\end{equation}
	Combining the two results completes the proof.
\end{proof}

\section{Coherence monotones from entanglement monotones\\
}
A quantum entanglement monotone of quantum measurements induces a quantum coherence monotone of quantum measurements. We require that a quantum coherence monotone $ C $ satisfies the following conditions:
\begin{enumerate}
	\item $ C(\mathcal{N}_{A}) \ge 0$; $ C(\mathcal{N}_{A}) = 0$ if and only if $ \mathcal{N}_{A}\in \mathbf{I}(d,n) $,
	\item $ C(\mathcal{S}_{R}\circ \mathcal{N}_{A}\circ \mathcal{F}_{A})\le C(\mathcal{N}_{A}) $ for any pre-processing channel $ \mathcal{F}_{A} \in \mathbf{UDI} $ and a classical post-processing channel $ \mathcal{S}_{R} $,
	\item $ C\left (\sum_{i}p_{i}\mathcal{N}_{A}^{(i)}\right )\le \sum_{i}p_{i} C\left (\mathcal{N}_{A}^{(i)}\right ) $.
\end{enumerate}
We require similar conditions for a quantum entanglement monotone $ E $ as well:
\begin{enumerate}
	\item $ E(\mathcal{N}_{AB}) \ge 0$; $ E(\mathcal{N}_{AB}) = 0$ if and only if $ \mathcal{N}_{AB}\in \sepm(A\! : \! B) $,
	\item $ E(\mathcal{S}_{R}\circ \mathcal{N}_{AB}\circ \mathcal{F}_{AB})\le E(\mathcal{N}_{AB}) $ for any pre-processing channel $ \mathcal{F}_{AB} $ that does not generate quantum entanglement from $ \sepm(A\! : \! B) $ and a classical post-processing channel $ \mathcal{S}_{R} $ acting on the system $ A $ and $ B $,
	\item $ E\left (\sum_{i}p_{i}\mathcal{N}_{AB}^{(i)}\right )\le \sum_{i}p_{i} E\left (\mathcal{N}_{AB}^{(i)}\right ) $.
\end{enumerate}

The following result establishes the existence of the induced quantum coherence monotone for quantum measurements:
\begin{theorem}\label{thm: coh. monotone via ent. monotone}
	Let $ \mathcal{M}_{A}\in \mathbf{M}(d,n) $ be a quantum measurement. Let $ \mathcal{E}_{B} \in \mathbf{I}(d,n)$ be an incoherent measurement given by
	\begin{equation}\label{key}
		\mathcal{E}_{B}=\begin{cases}
			\{E_{0},\dots,E_{d-1},0,\dots,0 \} & n\ge d,\\
			\{E_{0},\dots,E_{n-2}, I_{B}-\sum_{x=0}^{n-2}E_{x} \} & n < d,
		\end{cases}
	\end{equation}
	where $ E_{x}=\ketbra{x}{x}_{B} $.
	For $ n>1 $, a quantum entanglement monotone $ E $ for quantum measurements induces a quantum coherence monotone for quantum measurements as follows:
	\begin{equation}\label{key}
		C(\mathcal{M}_{A}) \coloneqq \sup_{\mathcal{F}_{AB}\in \mathbf{UDI}} E(\mathcal{M}_{A}\otimes \mathcal{E}_{B}\circ \mathcal{F}_{AB}).
	\end{equation}
\end{theorem}
\begin{proof}
	We verify the condition for $ C $ being a quantum coherence monotone:
	\begin{enumerate}
		\item First, note that $ C(\cdot)\ge 0 $ due to $ E(\cdot)\ge 0 $. To show that $ C(\mathcal{N}_{A}) =0$ for $ \mathcal{N}_{A}\in \mathbf{I}(d,n)$, $ \mathbf{I}(d\times d,n\times n)\subset \sepm(A\!:\!B) $ proves the ``if'' direction, while Theorem \ref{thm: meas. coh. ent. equality} assures the other direction.
		\item For any $ \mathcal{F}_{A}\in \mathbf{UDI} $ and a classical post-processing channel $ \mathcal{S}_{R} $ acting on the system $ A $ and $ B $, the monotonicity holds as follows:
		\begin{align}
			&C(\mathcal{S}_{R}\circ \mathcal{N}_{A}\circ \mathcal{F}_{A}) \nonumber\\
			&=\sup_{\mathcal{G}_{AB}\in \mathbf{UDI}} E(\left(\mathcal{S}_{R}\circ \mathcal{N}_{A}\circ \mathcal{F}_{A} \right)\otimes \mathcal{E}_{B}\circ \mathcal{G}_{AB})\\
			&\le\sup_{\mathcal{F}_{AB}'\in \mathbf{UDI}} E(\mathcal{S}_{R}\otimes \mathsf{Id}_{B} \circ \mathcal{N}_{A}\otimes \mathcal{E}_{B}\circ \mathcal{F}_{AB}')\\
			&\le\sup_{\mathcal{F}_{AB}'\in \mathbf{UDI}} E(\mathcal{N}_{A}\otimes \mathcal{E}_{B}\circ \mathcal{F}_{AB}')\\
			&= C(\mathcal{N}_{A}),
		\end{align}
		where we used the monotonicity of $ E $ and that $ \mathcal{F}_{A}\otimes \mathsf{Id}_{B}\circ \mathcal{G}_{AB}\in \mathbf{UDI} $ for $ \mathcal{G}_{AB}\in \mathbf{UDI} $.
		\item The convexity of the dynamic coherence monotone can be seen as below:
		\begin{align}
			&C\left (\sum_{i}p_{i}\mathcal{N}_{A}^{(i)}\right )\nonumber\\
			&= E\left (\sum_{i}p_{i} \mathcal{N}_{A}^{(i)}\otimes \mathcal{E}_{B}\circ \mathcal{F}_{AB}^{\ast}\right )\\
			&\le \sum_{i}p_{i}E\left (\mathcal{N}_{A}^{(i)}\otimes \mathcal{E}_{B}\circ \mathcal{F}_{AB}^{\ast}\right )\\
			&\le \sum_{i}p_{i} \sup_{\mathcal{F}_{AB}^{(i)}\in \mathbf{UDI}}E\left (\mathcal{N}_{A}^{(i)}\otimes \mathcal{E}_{B}\circ \mathcal{F}_{AB}^{(i)}\right )\\
			&\le \sum_{i}p_{i}C\left (\mathcal{N}_{A}^{(i)}\right ),
		\end{align}
		where we assumed and used the convexity of $ E $ in the first inequality.
	\end{enumerate}
\end{proof}
We finally remark that a single outcome measurement ($ n=1 $) is the trivial measurement $ \mathcal{M}_{A}=\{I_{A}\} $ that does not have any quantum resources.

\end{document}